\long\def\LongVersion#1\LongVersionEnd{#1}
\long\def\ShortVersion#1\ShortVersionEnd{}

\long\def\ignoreabs#1\ignoreabstractEnd{}

\LongVersion %{
\documentclass[11pt]{article}
\LongVersionEnd %}

\usepackage{amsmath, amssymb}

\usepackage{amsthm}

\usepackage{ifthen}

\usepackage{ifpdf}

\LongVersion %{
\ifpdf
\usepackage[pdftex]{graphicx}
\usepackage[pdftex]{hyperref}
\else
\usepackage[dvips]{graphicx}
\fi
\LongVersionEnd %}

\ifthenelse{\isundefined{\NoGenerateLetterSize}}
{
\ifpdf
\setlength{\pdfpagewidth}{8.5in}
\setlength{\pdfpageheight}{11in}
\else

\fi
}
{}

\newcommand{\Comment}[1]{}

\LongVersion %{
\LongVersionEnd %}

\newtheorem{theorem}{Theorem}[section]
\newtheorem{lemma}[theorem]{Lemma}

\newtheorem{proposition}[theorem]{Proposition}

\newtheorem*{observation*}{Observation}

\theoremstyle{definition}

\newtheorem*{definition*}{Definition}
\theoremstyle{plain}

\newcommand{\Vertices}[0]{\mathit{V}}
\newcommand{\Edges}[0]{\mathit{E}}
\newcommand{\Peers}[0]{\mathcal{P}}
\newcommand{\ERDC}[0]{\mathrm{ERDC}}
\newcommand{\PDDC}[0]{\mathrm{PDDC}}
\newcommand{\SPDDC}[0]{\mathrm{SPDDC}}
\newcommand{\FDC}[0]{\mathrm{FDC}}
\newcommand{\CalS}[0]{\mathcal{S}}

\newcounter{smallitemizec}
\newenvironment{smallitemize}
{   \setcounter{smallitemizec}{0}
    \vspace{-0.5ex}
  \begin{list}{$\bullet$}
    {\usecounter{smallitemizec}
      \setlength{\parsep}{0pt}
      \setlength{\itemsep}{0pt}}
    }{ \end{list}
   \vspace{-0.5ex}
}

\begin{document}

\LongVersion %{
\author{
Yuval Emek
\thanks{Computer Engineering and Networks Laboratory, ETH Zurich, Zurich,
Switzerland.
E-mail: {\tt yuval.emek@tik.ee.ethz.ch}.}
\and
Pierre Fraigniaud
\thanks{CNRS and University of Paris Diderot, France.
E-mail: {\tt pierref@liafa.jussieu.fr}.}
\and
Amos Korman
\thanks{CNRS and University of Paris Diderot, France.
E-mail: {\tt amos.korman@liafa.jussieu.fr}.}
\and
Shay Kutten
\thanks{Information Systems Group, Faculty of IE\&M, The Technion,
Haifa, 32000 Israel. E-mail: {\tt kutten@ie.technion.ac.il}.}
\and
David Peleg
\thanks{
Department of Computer Science and Applied Mathematics, The Weizmann
Institute of Science, Rehovot, 76100 Israel.
E-mail: {\tt david.peleg@weizmann.ac.il}.}
}
\LongVersionEnd %}

\title{Notions of Connectivity in Overlay Networks
\thanks{Supported by a France-Israel cooperation grant
(``Mutli-Computing'' project)
from the France Ministry of Science and Israel Ministry of Science. Supported by the ANR projects DISPLEXITY and PROSE, and by the INRIA project GANG}
}

%\date{\today}

\maketitle

\begin{abstract}
``How well connected is the network?''
This is one of the most fundamental questions one would ask when facing
the challenge of designing a communication network.
Three major notions of connectivity have been considered in the
literature, but in the context of traditional (single-layer) networks, they
turn out to be equivalent.
This paper introduces a model for studying the three notions of connectivity
in  {\em multi-layer} networks.
Using this model, it is easy to demonstrate that in multi-layer networks the
three notions may differ dramatically.
Unfortunately, in contrast to the single-layer case, where the values of the
three connectivity notions can be computed efficiently, it has been recently
shown in the context of WDM networks (results that can be easily translated to
our model) that the values of two of these notions of connectivity are hard to
compute or even approximate in multi-layer networks.
The current paper shed some positive light into the multi-layer connectivity
topic:
we show that the value of the third connectivity notion can be computed in
polynomial time and develop an approximation for the construction of well
connected overlay networks.
\end{abstract}

%\setcounter{page}{0}
%\thispagestyle{empty}
%\clearpage

%%%%%%%%%%%%%%%%%%%%%%%%%%%%%%%%%%%%%%%%%%%%%%%%%%%%%%%%%%%%%%%%%%%%%%%%%%%%%%
\section{Introduction}
%%%%%%%%%%%%%%%%%%%%%%%%%%%%%%%%%%%%%%%%%%%%%%%%%%%%%%%%%%%%%%%%%%%%%%%%%%%%%%
%%%%%%%%%%%%%%%%%%%%%%%%%%%%%%%%%%%%%%%
\subsection{Background and motivation}
%%%%%%%%%%%%%%%%%%%%%%%%%%%%%%%%%%%%%%%

The term ``connectivity'' in networks has more than one meaning,
but these meanings are equivalent in ``traditional'' networks.
While the graph theoretic definition of connectivity refers to the ability to
reach every node from every other node (a.k.a.\ \emph{$1$-connectivity}), the
term connectivity is often related also to the \emph{survivability} of a
network, namely, the ability to preserve $1$-connectivity whenever some links
fail.\footnote{
The current paper does not deal with {\em node} connectivity.
}
In other words, a network $G$ is said to be $k$-{\em connected} if it
satisfies the following ``connectivity property'' (CP):
\begin{description}
\item[(CP1)] $G$ remains connected whenever up to $k-1$ links are erased from
it.
\end{description}

However, there are also other meanings to connectivity.
A network $G$ is also said to be $k$-{\em connected} if it
satisfies the following  ``connectivity property'':
\begin{description}
\item[(CP2)] There exist $k$ pairwise {\em link-disjoint paths} from $s$ to
$t$ for every two nodes $s, t$ in $G$.
\end{description}
The equivalence of these two properties \cite{Meng27} enables numerous
practical applications.
For example, one of the applications of the existence of link-disjoint paths
(Property (CP2)) is to ensure survivability (Property (CP1)).
Often, a backup (link-disjoint) path is prepared in advance, and the traffic
is diverted to the backup path whenever a link on the primary path fails.
An example is the backup path protection mechanism in SONET networks (see,
e.g., \cite{protection}).

A third property capturing connectivity is based on the amount of flow that
can be shipped in the network between any source and any destination, defining
the capacity of a single link to be~$1$.
In other words, a network $G$ is said to be $k$-{\em connected} if it
satisfies the following  ``connectivity property'':
\begin{description}
\item[(CP3)] It is possible to ship $k$ units of flow from $s$ to $t$ for
every two nodes $s, t$ in $G$.
\end{description}
This property too is equivalent to the first two \cite{FF56}, and is also used
together with them.
For example, routing some flow of information around congestion (which may be
possible only if the network satisfies property (CP3) and thus can support
this additional flow) uses the second property, i.e., it relies on the
existence of multiple link-disjoint paths.

Current networks, however, offer multi-layered structures which yield
significant complications when dealing with the notion of connectivity.
In particular, the overlay/underlying network dichotomy plays a major role
in modeling communication networks, and overlay networks such as peer-to-peer
(P2P) networks, MPLS networks, IP/WDM networks, and SDH/SONET-WDM networks,
all share the same overall structure:
an \emph{overlay} network $H$ is implemented on top of an \emph{underlying}
network $G$.
This implementation can be abstracted as a \emph{routing scheme} that maps
overlay connections to underlying routes.
We comment that there are sometimes differences between such a mapping and the common notion of a routing scheme. Still, since the routing scheme often defines the mapping,
we shall term this mapping the
\emph{routing scheme}.

Often, the underlying network itself is implemented on top of yet another
network, thus introducing a multi-layer hierarchy.
Typically, the lower level underlying network is ``closer'' to the physical
wires, whereas the higher level network is a traffic network in which edges
capture various kinds of connections, depending on the context.
For the sake of simplicity, we focus on a pair of consecutive layers $G$ and
$H$.
This is sufficient to capture a large class of practical scenarios.

The current paper deals with what happens to the different connectivity
properties once we turn to the context of overlay networks.
As discussed later on, connectivity has been studied previously in the
``overlay network'' world under the ``survivability'' interpretation (CP1),
and it has been observed that, in this context, the connectivity parameter
changes, i.e., the connectivity of the overlay network may be different
from that of the underlying network.
Lee et al.\ \cite{prev} demonstrated the significance of this difference by
showing that the survivability property is computationally hard and even hard
to approximate in the multi-layer case.
Since the three aforementioned connectivity parameters may differ in
multi-layer networks (see Section~\ref{section:Model}), they also showed a
similar result for the disjoint paths connectivity property.

Interestingly, the motivation of Lee et al.\ for addressing the disjoint paths
connectivity property was the issue of flow.
One of the contributions of the current paper is to directly address this
issue, showing that in contrast to the previous two notions of connectivity,
the maximum flow supported by an overlay network can actually be computed in
polynomial time.

In the specific context of survivability, there has been other papers that
have shown  that the issue of connectivity in an overlay network is different
from that of connectivity in underlying networking.
Consider, for example,
a situation where several overlay edges (representing connections)
of $H$ pass over the same physical link.
Then all these overlay edges may be disconnected as a result of a
single hardware fault in that link, possibly disconnecting the overlay network.
The affected overlay links are said to {\em share} the {\em risk}
of the failure of the underlying physical link, hence they are referred to
in the literature as a \emph{shared risk link group} (SRLG).
An SRLG-based model for overlay networks was extensively studied in recent
years;\footnote{
Note that a common underlying link is not the only possible shared risk;
overlay links sharing a node may form a shared risk link group too.
}
see, e.g., \cite{intro-srlg} for a useful introduction to this notion and
\cite{cisco-srlg} for a discussion of this concept in the context of MPLS.
The SRLG model hides the actual structure of the underlying network, in the
sense that many different underlying networks can yield the same sets of
SRLG.
(For certain purposes, this is an advantage of the model.)
An even more general notion is that of \emph{Shared Risk Resource Group}
However, sometimes this model abstracts away too much information,
making certain computational goals (such as, e.g., flow computations)
harder to achieve.
%shay: maybe we no longer have the modivation to differentiate the models that much. We make another contribution to a rather well studied and motivated area. Hence, I may remove the discussion from the "related".
%
%The issue is discussed further in Section
%\ref{subsec:related}.

In contrast to the SRLG model, we present the alternative model of
\emph{deep connectivity}, which allows us to simultaneously consider all three
components: the overlay network, the underlying network, and the mapping (the
routing scheme).
Note that all three should be considered:
For example, if the underlying network is not connected, then neither can the
overlay network be.
The routing scheme also affects the connectivity properties as different
routing schemes may yield significantly different overlay link dependencies.
In some cases, routing is constrained to shortest paths, whereas in other
cases it can be very different.
In \emph{policy based} routing schemes (see, e.g. \cite{CiscoPolicy}), for
example, some underlying edges are not allowed to be used for routing from $u$
to $v$, which may cause the underlying path implementing the overlay link $(u,
v)$ to be much longer than the shortest $(u, v)$-path.

%%%%%%%%%%%%%%%%%%%%%%%%%%%%%%%%%%%%%%%
\subsection{The deep connectivity model}
\label{section:Model}
%%%%%%%%%%%%%%%%%%%%%%%%%%%%%%%%%%%%%%%

The underlying network is modeled by a (simple, undirected, connected)
graph $G$ whose vertex set $\Vertices(G)$ represents the network nodes,
and whose edge set $\Edges(G)$ represents the communication links between them.
Some nodes of the underlying network are designated as \emph{peers};
the set of peers is denoted by $\Peers \subseteq \Vertices(G)$.
The overlay network, modeled by a graph $H$, spans the peers, i.e.,
$\Vertices(H) = \Peers$ and $\Edges(H) \subseteq \Peers \times \Peers$;
$H$ typically represents a ``virtual'' network, constructed on top of
the peers in the underlying communication network $G$.

An edge $(u, v)$ in the overlay graph $H$ may not directly correspond to an
edge in the underlying graph $G$ (that is, $\Edges(H)$ is not necessarily a
subset of $\Edges(G)$).
Therefore, communication over a $(u, v)$ edge in $H$ should
often be routed along some multi-hop path connecting $u$ and $v$ in $G$.
This is the role of a \emph{routing scheme} $\rho : \Peers \times \Peers
\rightarrow 2^{\Edges(G)}$ that maps each pair $(u, v)$ of peers to some
simple path $\rho(u, v)$ connecting $u$ and $v$ in $G$.
A message transmitted over the edge $(u, v)$ in $H$ is physically disseminated
along the path $\rho(u, v)$ in $G$.
We then say that $(u, v)$ is \emph{implemented} by $\rho(u, v)$.
For the sake of simplicity, the routing scheme $\rho$ is assumed to be
symmetric, i.e., $\rho(u, v) = \rho(v, u)$.
More involved cases do exist in reality:
the routing scheme may be asymmetric, or may map some overlay edge into
multiple routes;
the simple model given here suffices to show interesting differences between
the various connectivity measures.

When a message is routed in $H$ from a peer $s \in \Peers$ to a
non-neighboring peer $t \in \Peers$ along some multi-hop path $\pi = (x_0, x_1,
\dots, x_k)$ with $x_0 = s$, $x_k = t$, and $(x_i, x_{i + 1}) \in E(H)$, it is
physically routed in $G$ along the concatenated path $\rho(x_0, x_1) \rho(x_1,
x_2) \cdots \rho(x_{k - 1}, x_k)$.
In some cases, when the overlay graph $H$ is known, it will be convenient to
define the routing scheme over the edges of $H$, rather than over all peer
pairs.

The notion of \emph{deep connectivity} grasps the
connectivity in the overlay graph $H$, while taking into account its
implementation by the underlying paths in $G$.
Specifically, given two peers $s, t \in \Peers$, we are interested in three
different parameters, each capturing a specific type of connectivity.
In order to define these parameters, we extend the domain of $\rho$ from
vertex pairs in $\Peers \times \Peers$ to collections of such pairs in the
natural manner, defining $\rho(F) = \bigcup_{(u, v) \in F} \rho(u, v)$ for
every $F \subseteq \Peers \times \Peers$.
In particular, given an $(s, t)$-path $\pi$ in $H$, $\rho(\pi) = \bigcup_{e
\in \pi} \rho(e)$ is the set of underlying edges used in the implementation of
the overlay edges along $\pi$.
\begin{itemize}

\item
The \emph{edge-removal deep connectivity} of $s$ and $t$ in $H$ with respect
to $G$ and $\rho$, denoted by $\ERDC_{G, \rho}(s, t, H)$, is defined as the size
of the smallest subset $F \subseteq \Edges(G)$ that intersects with
$\rho(\pi)$ for every $(s, t)$-path $\pi$ in $H$;
namely, the minimum number of underlying edges that should be removed in order
to disconnect $s$ from $t$ in the overlay graph.

\item
The \emph{path-disjoint deep connectivity} of $s$ and $t$ in $H$ with
respect to $G$ and $\rho$, denoted by $\PDDC_{G, \rho}(s, t, H)$, is defined as
the size of the largest collection $C$ of $(s, t)$-paths in $H$ such that
$\rho(\pi) \cap \rho(\pi') = \emptyset$ for every $\pi, \pi' \in C$ with $\pi
\neq \pi'$, i.e., the maximum number of overlay paths connecting $s$ to $t$
whose underlying implementations are totally independent of each other.

\item
The \emph{flow deep connectivity} of $s$ and $t$ in $H$ with respect to $G$
and $\rho$, denoted by $\FDC_{G, \rho}(s, t, H)$, is defined as the maximum
amount of flow\footnote{
In the setting of undirected graphs, flow may be interpreted in two different
ways depending on whether two flows along the same edge in opposite directions
cancel each other or add up.
Here, we assume the latter interpretation which seems to be more natural in
the context of overlay networks.
} that can be pushed from $s$ to $t$ in $G$ restricted to the
images under $\rho$ of the $(s, t)$-paths in $H$, assuming that each edge in
$\Edges(G)$ has a unit capacity.
Intuitively, if $s$ and $t$ are well connected, then it should be possible to
push a large amount of flow between them.

\end{itemize}

\par
\noindent {\bf Example:}
To illustrate the various definitions,
consider the underlying network $G$ depicted in Fig. \ref{fig:conn-defs}(a),
and the overlay network $H$ depicted in Fig. \ref{fig:conn-defs}(b).
The routing scheme $\rho$ assigns each of the 6 overlay edges adjacent to
the extreme $S$ and $T$ a simple route consisting of the edge itself.
For the remaining three overlay edges, the assigned routes are as follows:
\begin{eqnarray*}
\rho(U_1,U_4) & = & (U_1,M_1,M_2,U_2,U_3,U_4)\\
\rho(M_1,M_4) & = & (M_1,M_2,D_2,D_3,M_3,M_4)\\
\rho(D_1,D_4) & = & (D_1,D_2,D_3,U_3,U_4,D_4)~.
\end{eqnarray*}
The route $\rho(M_1,M_4)$ is illustrated by the dashed line in
Fig. \ref{fig:conn-defs}(a).
Note that in the original (underlying) network $G$, the connectivity of the
extreme nodes $S$ and $T$ is 3 under all three definitions. In contrast,
the values of the three connectivity parameters for the extreme nodes
$S$ and $T$ in the overlay network $H$ under the routing scheme $\rho$
are as follows:

\begin{itemize}
\item
The edge-removal deep connectivity of $s$ and $t$ in $H$
w.r.t. $G$ and $\rho$ is $\ERDC_{G, \rho}(s, t, H)=2$.
\\
Indeed, disconnecting the underlying edge $(D_2,D_3)$
plus any edge of the upper underlying route will disconnect $S$ from $T$.
\item
The path-disjoint deep connectivity of $s$ and $t$ in $H$
w.r.t. $G$ and $\rho$ is $\PDDC_{G, \rho}(s, t, H)=1$.
\\
Indeed, any two of the three overlay routes connecting $S$ and $T$
share an underlying edge.
\item
The \emph{flow deep connectivity} of $s$ and $t$ in $H$
w.r.t. $G$ and $\rho$ is $\FDC_{G, \rho}(s, t, H) = 3/2$.
\\
This is obtained by pushing 1/2 flow unit through each of the three
overlay routes.
\end{itemize}

\begin{figure}
\begin{center}
\includegraphics[width=0.95\linewidth]{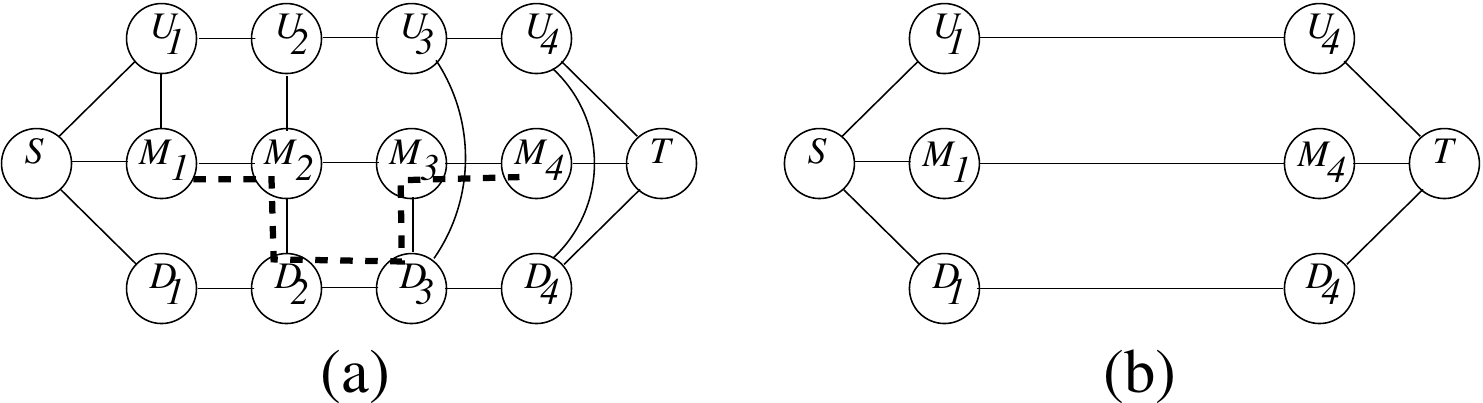}
\end{center}
\caption{\label{fig:conn-defs}
(a) The underlying graph $G$.
(b) The overlay network $H$.}
\vspace*{-.5cm}
\end{figure}

For each deep connectivity $(s, t)$-parameter $X_{G,
\rho}(s,t,H)$, we define the corresponding
\emph{all-pairs} variant $X_{G, \rho}(H) = \min_{s, t \in \Peers} X_{G,
\rho}(s, t, H)$.
When $G$ and $\rho$ are clear from the context, we may remove them from the
corresponding subscripts.

%%%%%%%%%%%%%%%%%%%%%%%%%%%%%%%%%%%%%%%
\subsection{Our contributions}
%%%%%%%%%%%%%%%%%%%%%%%%%%%%%%%%%%%%%%%

Our model for overlay networks makes it convenient to explore the
discrepancies between the different deep connectivity notions.
Classical results from graph theory, e.g., the fundamental min-cut max-flow
theorem \cite{FF56,EFS56} and Menger's theorem \cite{Meng27} state that the
three connectivity parameters mentioned above are equivalent when a single
layer network is considered.
Polynomial time algorithms that compute these parameters (in a single
layer network) were discovered early \cite{FF56,Dini70,EK72} and have since
become a staple of algorithms textbooks \cite{CLRS09,KT05}.
As mentioned above, previous results \cite{prev,CDPRV07}, when translated to
our model, have shown that in multi-layer networks, two of the three
connectivity parameters are computationally hard and even hard to
approximate.
Our first technical contribution is to expand on these negative
results by showing that the the path-disjoint deep connectivity property
cannot be approximated to any finite ratio when attention is restricted to
simple paths in the underlying graph.

On the positive side, we show that the flow deep connectivity parameters can
be computed in polynomial time (Section~\ref{section:EfficientAlgorithmFDC}),
thus addressing the motivation of \cite{prev} for studying the disjoint paths
property in overlay networks.
Then, we address the issue of constructing a ``good'' overlay graph for a
given underlying graph and routing scheme.
As opposed to the difficulty of approximating the value of the parameters, we
show that the related construction problem can sometimes be well
approximated.
Specifically, in Section~\ref{section:SparseOverlayGraphs}, we investigate the
problem of constructing $2$-edge removal deeply connected overlay graphs with
as few as possible (overlay) edges.
This problem is shown to be NP-hard, but we show that a
logarithmic-approximation for it can be obtained in polynomial-time.
We also devise a 2-approximation algorithm for particular, yet natural,
instances of the problem.

%%%%%%%%%%%%%%%%%%%%%%%%%%%%%%%%%%%%%%%%%%%%%%%%%%%%%%%%%%%%%%%%%%%%%%%%%%%%%%
\section{Hardness of approximation}
\label{section:HardnessApproximation}
%%%%%%%%%%%%%%%%%%%%%%%%%%%%%%%%%%%%%%%%%%%%%%%%%%%%%%%%%%%%%%%%%%%%%%%%%%%%%%

As mentioned earlier, Lee et al.~\cite{prev} established hardness of
approximation results for the problems of computing the parameters $\ERDC_{G,
\rho}(s, t, H)$, $\PDDC_{G, \rho}(s, t, H)$, and $\ERDC_{G, \rho}(H)$.
For completeness, we observe that the all-pairs variant $\PDDC_{G, \rho}(H)$
of the path-disjoint deep connectivity parameter is also hard to approximate,
establishing the following theorem, which is essentially a corollary of
Theorem~4 in \cite{prev} combined with a result of H{\aa}stad~\cite{Hast99}.

\begin{theorem} \label{theorem:HardnessAllPairsPDDC}
Unless NP = ZPP, the problem of computing the parameter $\PDDC_{G, \rho}(H)$
cannot be approximated to within a ratio of $|\Edges(H)|^{1 / 2 - \epsilon}$
for any $\epsilon > 0$.
\end{theorem}

We now turn to show that the following natural variants of the $\PDDC$
parameters cannot be approximated to within any finite ratio.
Let $\SPDDC_{G, \rho}(s, t, H)$ denote the size of the largest collection $C$
of $(s, t)$-paths in $H$ such that all paths $\pi \in C$ are implemented by
simple paths $\rho(\pi)$ in $G$ and $\rho(\pi) \cap \rho(\pi') = \emptyset$
for every $\pi, \pi' \in C$ with $\pi \neq \pi'$;
let $\SPDDC_{G, \rho}(H) = \min_{s, t \in \Peers} \SPDDC_{G, \rho}(s, t, H)$.
Note that these parameters are merely a restriction of the $\PDDC$ parameters
to simple paths (hence the name, which stands for \emph{simple} path-disjoint
deep connectivity).

The inapproximability of the $\SPDDC_{G, \rho}(s, t, H)$ parameter is proved
by reducing the set packing problem to the problem of distinguishing between
the case $\SPDDC_{G, \rho}(s, t, H) = 0$ and the case $\SPDDC_{G, \rho}(s, t,
H) \geq 1$.
In fact, the vertices $s, t \in \Vertices(H)$ that minimize $\SPDDC_{G,
\rho}(s, t, H)$ in this reduction are known in advance, thus establishing the
impossibility of approximating the all-pairs parameter $\SPDDC_{G, \rho}(H)$
as well.
The proofs of the following two theorems are deferred to
Appendix~\ref{appendix:Hardness}.

\begin{theorem} \label{theorem:HardnessSPDDC}
Unless P = NP, the problem of computing the parameter $\SPDDC_{G, \rho}(s,
t, H)$ cannot be approximated to within any finite ratio.
\end{theorem}

\begin{theorem} \label{theorem:HardnessAllPairsSPDDC}
Unless P = NP, the problem of computing the parameter $\SPDDC_{G, \rho}(H)$
cannot be approximated to within any finite ratio.
\end{theorem}

%%%%%%%%%%%%%%%%%%%%%%%%%%%%%%%%%%%%%%%%%%%%%%%%%%%%%%%%%%%%%%%%%%%%%%%%%%%%%%
\section{Efficient algorithm for $\FDC$}
\label{section:EfficientAlgorithmFDC}
%%%%%%%%%%%%%%%%%%%%%%%%%%%%%%%%%%%%%%%%%%%%%%%%%%%%%%%%%%%%%%%%%%%%%%%%%%%%%%

In this section we develop a polynomial time algorithm that computes the
flow deep connectivity parameter $\FDC(s, t, H)$ (which clearly provides
an efficient computation of the parameter $\FDC(H)$ as well).
Consider some underlying graph $G$, routing scheme $\rho$, overlay graph $H$,
and two vertices $s, t \in \Vertices(H)$.
Let $\mathcal{P}$ denote the collection of all simple $(s, t)$-paths in $H$.
For each path $p \in \mathcal{P}$ and for each edge $e \in \Edges(G)$, let
$\psi(p, e)$ be the number of appearances of the edge $e$ along the image of
$p$ under $\rho$.
We begin by representing the parameter $\FDC(s, t, H)$ as the outcome of the
following linear program:
$$
\begin{array}{llr}
\max & \sum_{p \in \mathcal{P}} x_{p} \quad \text{s.t.} & \\
& \sum_{p \in \mathcal{P}} \psi(p, e) \cdot x_{p} \leq 1 & \forall e \in
\Edges(G) \\
& x_{p} \geq 0 & \forall p \in \mathcal{P}
\end{array}
$$

The variable $x_{p}$ represents the amount of flow pushed along the image
under $\rho$ of the path $p$ for every $p \in \mathcal{P}$.
The goal is to maximize the total flow pushed along the images of all paths in
$\mathcal{P}$ subject to the constraints specifying that the sum of flows
pushed through any edge is at most $1$.
This linear program may exhibit an exponential number of variables, so let us
consider its dual program instead:
$$
\begin{array}{llr}
\min & \sum_{e \in \Edges(G)} y_{e} \quad \text{s.t.} & \\
& \sum_{e \in \Edges(G)} \psi(p, e) \cdot y_{e} \geq 1 & \forall p \in
\mathcal{P} \\
& y_{e} \geq 0 & \forall e \in \Edges(G)
\end{array}
$$

The dual program can be interpreted as fractionally choosing as few as
possible edges of $G$ so that the image under $\rho$ of every path $p$ in
$\mathcal{P}$ traverses in total at least one edge.
We cannot solve the dual program directly as it may have an exponential
number of constraints.
Fortunately, it admits an efficient separation oracle, hence it can be solved
in polynomial time (see, e.g., \cite{GLS93}).

Given some real vector $\vec{y}$ indexed by the edges in $\Edges(G)$, our
separation oracle either returns a constraint which is violated by $\vec{y}$
or reports that all the constraints are satisfied and $\vec{y}$ is a feasible
solution.
Recall that a violated constraint corresponds to some path $p \in \mathcal{P}$
such that $\sum_{e \in \Edges(G)} \psi(p, e) \cdot y_{e} < 1$.
Therefore our goal is to design an efficient algorithm that finds such
a path $p \in \mathcal{P}$ if such a path exists.

Let $w(e) = \sum_{e' \in \rho(e)} y_{e'}$ for every edge $e \in \Edges(H)$ and
let $H'$ be the weighted graph obtained by assigning weight $w(e)$ to each
edge $e \in \Edges(H)$.
The key observation in this context is that the (weighted) length of an
$(s, t)$-path $p'$ in $H'$ equals exactly to $\sum_{e \in \Edges(G)} \psi(p,
e) \cdot y_{e}$, where $p$ is the path in $H$ that corresponds to $p'$ in
$H'$.
Therefore, our separation oracle is implemented simply by finding a shortest
$(s, t)$-path $p^{*}$ in $H'$:
if the length of $p^{*}$ is smaller than $1$, then $p^{*}$ corresponds to a
violated constraint;
otherwise, the length of all $(s, t)$-paths in $H'$ is at least $1$, hence
$\vec{y}$ is a feasible solution.
This establishes the following theorem.

\begin{theorem} \label{theorem:EfficientFDC}
The parameters $\FDC_{G, \rho}(s, t, H)$ and $\FDC_{G, \rho}(H)$ can be
computed in polynomial time.
\end{theorem}

%%%%%%%%%%%%%%%%%%%%%%%%%%%%%%%%%%%%%%%%%%%%%%%%%%%%%%%%%%%%%%%%%%%%%%%%%%%%%%
\section{Sparsest $2$-$\ERDC$ overlay graphs}
\label{section:SparseOverlayGraphs}
%%%%%%%%%%%%%%%%%%%%%%%%%%%%%%%%%%%%%%%%%%%%%%%%%%%%%%%%%%%%%%%%%%%%%%%%%%%%%%

In this section we are interested in the following problem, referred to as the
\emph{sparsest $2$-$\ERDC$ overlay graph} problem:
given an underlying graph $G$, a peer set $\Peers \subseteq \Vertices(G)$, and
a routing scheme $\rho : \Peers \times \Peers \rightarrow 2^{\Edges(G)}$,
construct the sparsest overlay graph $H$ for $\Peers$
(in terms of number of overlay edges) satisfying $\ERDC_{G, \rho}(H) \geq 2$.
Of course, one has to make sure that such an overlay graph $H$ exists, so in
the context of the sparsest $2$-$\ERDC$ overlay graph problem we always assume
that $\ERDC_{G, \rho}(K_{\Peers}) \geq 2$, where $K_{\Peers}$ is the complete graph
on $\Peers$.
This means that a trivial solution with ${n}\choose{2}$ edges, where $n =
|\Peers|$, always exists and the challenge is to construct a sparser one.

%%%%%%%%%%%%%%%%%%%%%%%%%%%%%%%%%%%%%%%
\subsection{Hardness}
%%%%%%%%%%%%%%%%%%%%%%%%%%%%%%%%%%%%%%%

We begin our treatment of this problem with a hardness result.

\begin{theorem} \label{theorem:HardnessSparsestERDC}
The sparsest $2$-$\ERDC$ overlay graph problem is NP-hard.
\end{theorem}
\begin{proof}
The assertion is proved by a reduction from the \emph{Hamiltonian path}
problem.
Consider an $n$-vertex graph $G_0$ input to the Hamiltonian path problem.
Transform it into an instance of the sparsest $2$-$\ERDC$ overlay graph
problem as follows:
Construct the underlying graph $G$ by setting $\Vertices(G) = \Vertices(G_0)
\cup \{x, y\}$ and $\Edges(G) = \Edges(G_0) \cup \{(x, y)\} \cup \{ (v, x),
(v, y) \mid v \in \Vertices(G_0) \}$ and let $\Peers = \Vertices(G_0)$.
Define the routing scheme $\rho$ by setting
$$
\rho(u, v) =
\left\{
\begin{array}{ll}
(u, v) & \text{ if } (u, v) \in \Edges(G_0) \\
(u, x, y, v) & \text{ otherwise.}
\end{array}
\right.
$$
This transformation is clearly polynomial in $n$.

We argue that $G_0$ admits a Hamiltonian path if and only if there exists an
overlay graph $H$ for $\Peers$ so that $|\Edges(H)| = n$ and $\ERDC_{G,
\rho}(H) \geq 2$.
To that end, suppose that $G_0$ admits a Hamiltonian path $\pi$.
If $\pi$ can be closed to a Hamiltonian cycle (in $G_0$), then take $H$ to be
this Hamiltonian cycle.
Otherwise, take $H$ to be the cycle consisting of $\pi$ and a virtual edge
connecting between $\pi$'s endpoints.
In either case, $H$ clearly has $n$ edges and by the design of
$\rho$, $H$ satisfies $\ERDC_{G, \rho}(H) = 2$.

Conversely, if there exists an overlay graph $H$ for $\Peers$ so that
$|\Edges(H)| = n$ and $\ERDC_{G, \rho}(H) \geq 2$, then $H$ must form a
Hamiltonian cycle $C$ in $\Peers \times \Peers$.
This cycle can contain at most one virtual edge as otherwise, the removal of
$(x, y)$ breaks two edges of $C$ which means that $\ERDC_{G, \rho}(H) < 2$.
By removing this virtual edge, we are left with a Hamiltonian path in $G_0$.
\end{proof}

%%%%%%%%%%%%%%%%%%%%%%%%%%%%%%%%%%%%%%%
\subsection{Constructing sparse $2$-$\ERDC$ overlay graphs}
%%%%%%%%%%%%%%%%%%%%%%%%%%%%%%%%%%%%%%%

On the positive side, we develop a polynomial time logarithmic approximation
algorithm for the sparsest $2$-$\ERDC$ overlay graph problem.
Our algorithm proceeds in two stages:
First, we take $T$ to be an arbitrary spanning tree of $\Peers \times
\Peers$.
Subsequently, we aim towards (approximately) solving the following
optimization problem, subsequently referred to as the \emph{overlay
augmentation} problem:
augment $T$ with a minimum number of $\Peers \times \Peers$ edges so that
the resulting overlay graph $H$ satisfies $\ERDC_{G, \rho}(H) \geq 2$.

We will soon explain how we cope with this optimization problem, but first let
us make the following observation.
Denote the edges in $\rho(T)$ by $\rho(T) = \{ e_1, \dots, e_{\ell} \}$ and
consider some overlay graph $H$ such that $\Edges(H) \supseteq T$ and some $1
\leq i \leq \ell$.
Let $F_{i}(H)$ be the collection of connected components of the graph obtained
from $H$ by removing all edges $e \in \Edges(H)$ such that $e_i \in \rho(e)$.
Fix
$$
\kappa_{i}(H) = |F_{i}(H)| - 1
\quad \text{and} \quad
\kappa(H) = \sum_{i = 1}^{\ell}
\kappa_{i}(H) \, .
$$
We think of $\kappa(H)$ as a measure of the distance of the overlay graph
$H$ from being a feasible solution to the overlay augmentation problem (i.e.,
$\ERDC(H) \geq 2$).

\begin{observation*}
An overlay graph $H \supseteq T$ satisfies $\ERDC(H) \geq 2$ if and only if
$\kappa(H) = 0$.
\end{observation*}
\begin{proof}
If $\ERDC(H) \geq 2$, then $F_{i}(H)$ must consist of a single connected
component for every $1 \leq i \leq \ell$, thus $\kappa(H) = 0$.
Conversely, if $\kappa(H) = 0$, then necessarily $|F_{i}(H)| = 1$ for every $1
\leq i \leq \ell$, which means that $H$ does not disconnect by the removal of
any edge $e_i \in \rho(T)$.
It is also clear that the removal of any edge in $\Edges(G) - \rho(T)$ does
not disconnect $H$ as the tree $T$ remains intact.
Therefore, $\ERDC(H) \geq 2$.
\end{proof}

Consider some edge $e \in (\Peers \times \Peers) - \Edges(H)$ and let $H \cup
\{e\}$ denote the overlay graph obtained from $H$ by adding the edge $e$.
By the definition of the parameter $\kappa$, we know that
$
\Delta_{i}(e, H) = \kappa_{i}(H) - \kappa_{i}(H \cup \{e\})
$
is either $0$ or $1$ for any $1 \leq i \leq \ell$.
Fixing
$
\Delta(e, H) = \kappa(H) - \kappa(H \cup \{e\})
$,
we observe that:
$\Delta(e, H) = \sum_{i = 1}^{\ell} \Delta_{i}(e, H)$
referred to as Property $\mathbf{(\star)}$.

We are now ready to complete the description of our approximation algorithm.
Starting from $H = T$, the algorithm gradually adds edges to $H$ as long as
$\kappa(H) > 0$ according to the following greedy rule.
At any intermediate step, we add to $H$ an edge $e \in (\Peers \times \Peers)
- \Edges(H)$ that yields the maximum $\Delta(e, H)$.
When $\kappa(H)$ decreases to zero, the algorithm terminates (recall that this
means that $\ERDC(H) \geq 2$).

The analysis of our approximation algorithm relies on the following
proposition.

\begin{proposition} \label{proposition:Submodularity}
The parameter $\kappa$ can be computed in polynomial time.
Moreover, for every two overlay graphs $H_1, H_2$ such that $\Edges(H_1)
\subseteq \Edges(H_2)$, we have
\begin{smallitemize}
\item[\rm (1)]
$\kappa(H_1) \geq \kappa(H_2)$; and
\item[\rm (2)]
$\Delta(e, H_1) \geq \Delta(e, H_2)$ for every edge $e \in \Peers \times
\Peers$.
\end{smallitemize}
\end{proposition}
\begin{proof}
The fact that $\kappa$ can be computed efficiently and the fact that
$\kappa(H_1) \geq \kappa(H_2)$ are clear from the definition of $\kappa$, so
our goal is to prove that $\Delta(e, H_1) \geq \Delta(e, H_2)$ for every $e
\in \Peers \times \Peers$.
By Property $\mathbf{(\star)}$, it suffices to show that
$\Delta_{i}(e, H_1) \geq \Delta_{i}(e, H_2)$ for every $1 \leq i \leq \ell$.
If $\Delta_{i}(e, H_2) = 0$, then this holds vacuously, so suppose that
$\Delta_{i}(e, H_2) = 1$.
This means that $e_i \notin \rho(e)$ and the endpoints of $e$ belong to
different connected components in $F_{i}(H_2)$.
But since $\Edges(H_1) \subseteq \Edges(H_2)$, it follows that the endpoints
of $e$ must also belong to different connected components in $F_{i}(H_1)$,
hence $\Delta_{i}(e, H_1) = 1$ as well.
\end{proof}

Proposition~\ref{proposition:Submodularity} implies that the overlay
augmentation problem falls into the class of \emph{submodular cover} problems
(cf. \cite{W82,BKP01}) and our greedy approach is guaranteed to have an
approximation ratio of at most $\ln(\kappa(T)) + 1 = O (\log N)$, where $N =
|\Vertices(G)|$.
More formally, letting $\hat{H}$ be a sparsest overlay graph such that
$\hat{H} \supseteq T$ and $\ERDC(\hat{H}) \geq 2$, it is guaranteed that the
overlay graph $H$ generated by our greedy approach satisfies
$|\Edges(H) - T| \leq O (\log N) \cdot |\Edges(\hat{H}) - T|$.

To conclude the analysis, let $H^{*}$ be an optimal solution to the sparsest
$2$-$\ERDC$ overlay graph problem.
Clearly, $|\Edges(H^{*})| > n - 1 = |T|$.
Moreover, since $\Edges(H^{*}) \cup T$ is a candidate for $\hat{H}$, it
follows that $|\Edges(H^{*}) \cup T| \geq |\hat{H}|$, thus $|\Edges(H^{*})|
\geq |\Edges(\hat{H}) - T|$.
Therefore,
\begin{align*}
|\Edges(H)|
~ \leq & ~
O (\log N) \cdot |\Edges(\hat{H}) - T| + |T| \\
\leq & ~
O (\log N) \cdot |\Edges(H^{*})| + |\Edges(H^{*})| \\
~ = & ~
O (\log(N) \cdot |\Edges(H^{*})|)
\end{align*}
which establishes the following theorem.

\begin{theorem} \label{theorem:ApproximationSparsestERDC}
The sparsest $2$-$\ERDC$ overlay graph problem admits a polynomial-time
logarithmic-approximation.
\end{theorem}

%%%%%%%%%%%%%%%%%%%%%%%%%%%%%%%%%%%%%%%
\subsection{A special case}
%%%%%%%%%%%%%%%%%%%%%%%%%%%%%%%%%%%%%%%

We now turn to consider instances of the sparsest $2$-$\ERDC$ overlay graph
problem satisfying the following two simplifying assumptions:
\begin{smallitemize}
\item[(1)] all vertices are peers, namely, $\Peers = \Vertices(G)$; and
\item[(2)] $\rho$ maps every edge in $\Edges(G)$ to itself.
\end{smallitemize}
We show that under these assumptions, one can always find a feasible solution
with at most $2 n - 2$ edges.
Since any overlay graph $H$ satisfying $\ERDC(H) \geq 2$ must have at least
$n$ edges, this immediately implies a $2$-approximation for our problem.

\LongVersion %{
The construction proceeds, once again, in two stages:
First, we take $H$ to be an arbitrary spanning tree $T$ of $G$.
Then, for each edge $e \in T$, we add to $H$ some edge $e' \in
\Edges(G)$ such that $e$ belongs to the cycle closed by appending $e'$ to $T$.
Note that such an edge $e'$ exists since $G$ is $2$-edge connected (as
otherwise, $\ERDC_{G, \rho}(K_{\Peers}) < 2$).
The overlay graph $H$ clearly contains at most $2 n - n$ edges.
Moreover, assumption~(A2) implies that $\rho$ maps every edge in $\Edges(H)$
to itself.
Therefore, our construction guarantees that $\ERDC(H) \geq 2$:
on one hand, the spanning tree $T$ ensures that the graph obtained from $H$ by
removing the edge $e$ is connected for every edge $e \notin T$;
on the other hand, $H$ does not disconnect by the removal of any single edge
$e \in T$ due to the second stage of the construction.
\LongVersionEnd %}

\begin{theorem} \label{theorem:ConstructionSparseERDC}
Under assumptions~(1) and~(2), an overlay graph $H$ with $2 n - 2$ edges
satisfying $\ERDC(H) \geq 2$ can be constructed in polynomial time.
\end{theorem}

\newpage
%%%%%%%%%%%%%%%%%%%%%%%%%%%%%%%%%%%%%%%%%%%%%%%%%%%%%%%%%%%%%%%%%%%%%%%%%%%%%%

%%%%%%%%%%%%%%%%%%%%%%%%%%%%%%%%%%%%%%%%%%%%%%%%%%%%%%%%%%%%%%%%%%%%%%%%%%%%%%
\clearpage

\pagenumbering{roman}
\appendix

\renewcommand{\theequation}{A-\arabic{equation}}
\setcounter{equation}{0}
\begin{center}
\textbf{\large{APPENDIX}}
\end{center}

%%%%%%%%%%%%%%%%%%%%%%%%%%%%%%%%%%%%%%%%%%%%%%%%%%%%%%%%%%%%%%%%%%%%%%%%%%%%%%
\section{Hardness of approximation for the simple path variants}
\label{appendix:Hardness}
%%%%%%%%%%%%%%%%%%%%%%%%%%%%%%%%%%%%%%%%%%%%%%%%%%%%%%%%%%%%%%%%%%%%%%%%%%%%%%
Similarly to \cite{prev}, we base our hardness results on the relationship
between path-disjoint deep connectivity and encoding \emph{set systems}.
An \emph{$(m, n)$-set system} $\CalS$ is a pair $(D, X)$, where: 
\begin{smallitemize}
\item
$D$ is a domain of elements, $|D| = m$; and
\item 
$X \subseteq 2^D$ is a collection of sets in the domain $D$, $|X| = n$.
\end{smallitemize}
It is convenient to represent $\CalS$ as a Boolean characteristic function,
$\CalS : [m] \times [n] \rightarrow \{0, 1\}$, so that for every $i \in [m]$
and $j \in [n]$, $\CalS(i, j) = 1$ if and only if the $i^{\text{th}}$ element
of $D$ is included in the $j^{\text{th}}$ set of $X$.

\begin{lemma} \label{result:EncodeSetSystem}
For every overlay graph $H$, edge subset $F \subseteq \Edges(H)$ of
cardinality $|F| = n$, and $(m, n)$-set system $\CalS=(D,X)$, there exist an
underlying graph $G$ and a routing scheme $\rho : \Edges(H) \rightarrow
2^{\Edges(G)}$ such that: 
\begin{smallitemize}
\item[\rm (1)]
$|\Vertices(G)| = |\Vertices(H)| + O(m \cdot n)$;
\item[\rm (2)]
$\Edges(G) = (\Edges(H) - F) \cup E_{D} \cup E_{\rho}$, where
$E_D$ and $E_\rho$ are sets of new edges on $\Vertices(G)$, $|E_{D}| = m$, and
$|E_{\rho}| = O (m \cdot n)$;
\item[\rm (3)]
for every $e_j \in F$, $j \in [n]$, and $e_i \in E_{D}$, $i \in [m]$, it holds
that $e_i \in \rho(e_j) \iff \CalS(i, j) = 1$; and
\item[\rm (4)]
for every $e \in E_{\rho}$, there exists a unique $e_j \in F$, $j \in [n]$,
such that $e \in \rho(e_j)$.
\end{smallitemize}
\end{lemma}
\begin{proof}
Informally, $G$ and $\rho$ are designed so that the set system $\CalS=(D,X)$
is encoded in such a way that $E_{D}$ corresponds to the element domain $D$
and $F$ corresponds to the set collection $X$. For a more precise description 
of the design of the underlying graph $G$ and routing scheme $\rho$, 
let us first present a simpler construction where $G$ is allowed to be 
a multigraph (with parallel edges).

The vertex set of $G$ consists of the vertices of $H$ and two additional new
vertices $v_{i}^{a}, v_{i}^{b}$ for every $i \in [m]$.
The edge set of $G$ consists of three disjoint subsets.
The first subset is just the edges of $\Edges(H) - F$;
$\rho$ maps every such $H$-edge to itself, namely, the edge 
$(x, y) \in \Edges(H) - F$ is mapped to the path in $G$
that consists of the single edge $(x, y)$.
The second subset is $E_{D} = \{ e_i = (v_{i}^{a}, v_{i}^{b}) \mid i \in [m]\}$.
The routing scheme $\rho$ is designed to guarantee that an edge $e_i \in
E_{D}$ appears in the path $\rho(e_j)$, $e_j \in F$, if and only if $\CalS(i,
j) = 1$.
For that purpose, we introduce the third subset $E_{\rho}$:
for each $j \in [n]$, assuming that the edge $e_j \in F$ connects the vertices
$x$ and $y$ in $H$ and that the $j^{\text{th}}$ set in $\CalS$ is $\{ i_1,
\dots, i_k \} \subseteq [m]$, we add to $E_{\rho}$ the edge $(x,
v_{i_1}^{a})$, the edges $(v_{i_{\ell}}^{b}, v_{i_{\ell + 1}}^{a})$ for $\ell
= 1, \dots, k - 1$, and the edge $(v_{i_{k}}^{b}, y)$.
It is important to point out that a new copy of those edges are added to
$E_{\rho}$ for each $j \in [n]$, which may create edge multiplicities.
Finally, $\rho$ maps the edge $e_j$ to the path $p_j = \langle x, v_{i_1}^{a},
v_{i_1}^{b}, \dots, v_{i_k}^{a}, v_{i_k}^{b}, y \rangle$.
Refer to Figure~\ref{figure:EncodeSetSystem} for an illustration.

The construction of $G$ immediately implies that $|\Vertices(G)| =
|\Vertices(H)| + 2 m$ and that $|E_{D}| = m$.
To see that $|E_{\rho}| = O (n \cdot m)$, observe that each
$E_{\rho}$-edge belongs to some path $p_j$, $j \in [n]$, and that each such
path admits at most two hops for every $i \in [m]$.
It remains to show that for every $j \in [n]$ and $i \in [m]$, $e_i \in
\rho(e_j)$ if and only if $\CalS(i, j) = 1$, which follows directly
from the design of $\rho$ as the path $p_j$ goes through an edge $e_i \in
E_{D}$ if and only if $\CalS(i, j) = 1$.

\begin{figure}
\begin{center}
\includegraphics[width=0.6\linewidth]{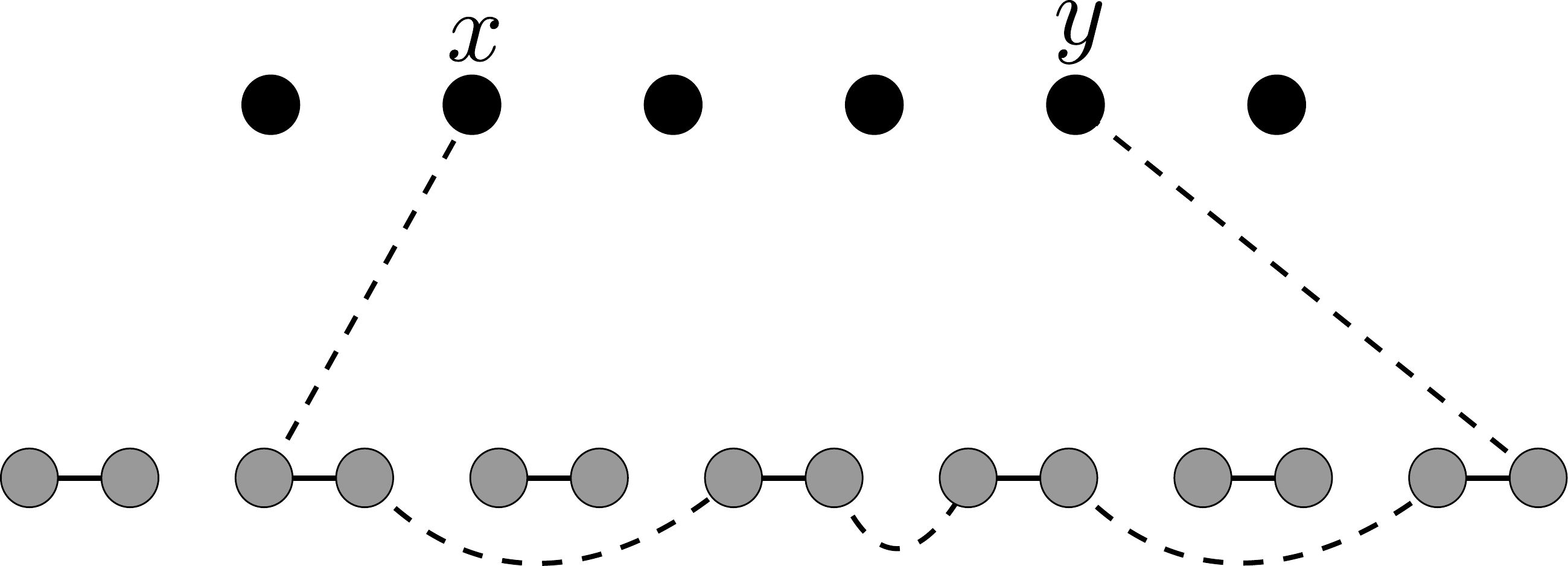}
\end{center}
\caption{\label{figure:EncodeSetSystem}
The underlying (multi)graph $G$.
The vertices of $H$ are depicted by the black circles;
the vertices in $\{ v_{i}^{a}, v_{i}^{b} \mid i \in [m] \}$ are depicted by
the gray circles.
The $E_{D}$-edges are depicted by the solid segments;
the $E_{\rho}$-edges along the path $p_j$ to which $\rho$ maps the edge $e_j =
(x, y) \in \Edges(H)$ are depicted by the dashed segments.
}
\vspace*{-.5cm}
\end{figure}

Finally, note that edge multiplicities in the constructed multigraph $G$ 
may occur only among the edges in $E_{\rho}$.
Hence, by subdividing each edge $e = (u, v) \in E_{\rho}$, replacing it with the
edges $(u, z_e)$ and $(v, z_e)$, where $z_e$ is a new vertex, we can turn the
multigraph $G$ into a simple graph at the cost of adding $O (m \cdot n)$ extra
vertices. The lemma follows.
\end{proof}

The hardness results established in this section are based on
Lemma~\ref{result:EncodeSetSystem} by a reduction from the \emph{set packing}
problem:
Given an $(m, n)$-set system $\CalS = (D, X)$ and a positive integer $k$, the
set packing problem asks whether there exist $k$ pairwise disjoint sets in
$X$.
The problem is known to be NP-complete;
in fact, it is among the original 21 NP-complete problems listed by
Karp~\cite{Karp72}.

The inapproximability of the $\SPDDC_{G, \rho}(s, t, H)$ parameter is proved
by reducing the set packing problem to the problem of distinguishing between 
the case
$\SPDDC_{G, \rho}(s, t, H) = 0$ and the case $\SPDDC_{G, \rho}(s, t, H) \geq 1$.
Given as input to the set packing problem some $(m, n)$-set system $\CalS$
with set collection $X = \{ S^1, \dots, S^n \}$ and positive integer $k$,
we first construct the $(m, k \cdot n)$-set system $\CalS'$ obtained
from $\CalS$ by creating $k$ identical copies $S_{1}^{j}, \dots, S_{k}^{j}$
of each set $S^j$ in $\CalS$.
Clearly, $\CalS'$ admits a set packing of size $k$ if and only if $\CalS$
admits a set packing of size $k$.
Then, we construct the overlay graph $H$ as illustrated in
Figure~\ref{figure:OverlayHops}.
Let $\{ F = (u_{\ell - 1}, v_{\ell}^{j}) \mid 1 \leq \ell \leq k, j \in [n] \}$ and
take $G$ and $\rho$ to be the underlying graph and routing scheme promised by
Lemma~\ref{result:EncodeSetSystem} when applied to $H$, $F$, and
$\CalS'$, where $\rho$ is organized so that edge $(u_{\ell - 1},
v_{\ell}^{j})$ in $F$ corresponds to set $S_{\ell}^{j}$ in $\CalS'$.

\begin{figure}
\begin{center}
\includegraphics[width=0.8\linewidth]{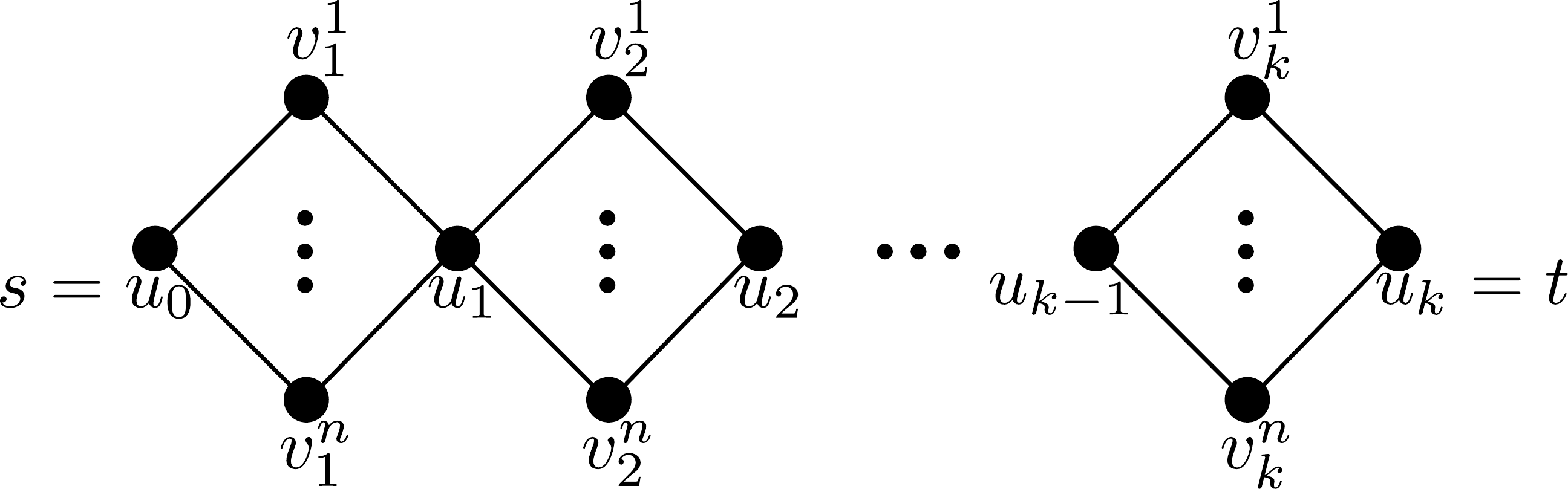}
\end{center}
\caption{\label{figure:OverlayHops}
The overlay graph $H$ used in the hardness proof of $\SPDDC_{G, \rho}(s, t, H)$.
}
\end{figure}

Assume first that $\SPDDC_{G, \rho}(s, t, H) \geq 1$ and let $\pi$ be an $(s,
t)$-path in $H$ such that its image under $\rho$ is a simple path in $G$.
Let $(u_0, v_{1}^{j_1}), \dots, (u_{k - 1}, v_{k}^{j_k})$ be the edges
traversed by $\pi$ on odd hops.
By the definition of $\CalS'$ and the design of $\rho$, it must be that
$j_{\ell} \neq j_{\ell'}$ for every $1 \leq \ell < \ell' \leq k$ as otherwise,
$\rho(\pi)$ does not form a simple path in $G$.
Lemma~\ref{result:EncodeSetSystem} then guarantees that the sets $S^{j_1},
\dots, S^{j_k}$ in $\CalS$ are pairwise disjoint.
In the converse direction, assume that the sets $S^{j_1}, \dots, S^{j_k}$ in
$\CalS$ are pairwise disjoint.
By the definition of $\CalS'$, the sets $S_{1}^{j_1}, \dots, S_{k}^{j_k}$ in
$\CalS'$ are also pairwise disjoint.
Lemma~\ref{result:EncodeSetSystem} then guarantees that the image under
$\rho$ of the path $\langle u_0, v_{1}^{j_1}, u_1, v_{2}^{j_2}, \dots, u_{k - 1},
v_{k}^{j_k}, u_k \rangle$ in $H$ is simple, hence $\SPDDC_{G, \rho}(s, t, H)
\geq 1$, which establishes Theorem~\ref{theorem:HardnessSPDDC}.
Theorem~\ref{theorem:HardnessAllPairsSPDDC} follows by observing that in the
aforementioned construction, $\SPDDC_{G, \rho}(x, y, H)$ is minimized by
taking $x = s$ and $y = t$.

\end{document}